\documentclass[11pt]{article}

\usepackage{times}
\usepackage{stmaryrd}
\usepackage{amssymb}
\usepackage{empheq}
\usepackage{verbatim}

\usepackage{datetime}
\usepackage{hyperref}

\usepackage{amsmath,amsfonts,amsthm}
\usepackage{graphicx}

\newenvironment{bzitemize}
  {\begin{list}
     {$\bullet$}
     {\setlength{\itemsep}{0ex}}}
  {\end{list}}

\newcommand{\argmaxlex}{\operatornamewithlimits{argmax^{\mathrm{lex}}}}
\newcommand{\argminlex}{\operatornamewithlimits{argmin^{\mathrm{lex}}}}

\newcommand{\maxlex}{\operatornamewithlimits{max^{\mathrm{lex}}}}
\newcommand{\minlex}{\operatornamewithlimits{min^{\mathrm{lex}}}}

\newcommand{\MinImprove}{\mathrm{Improve}_{\mathrm{Min}}}
\newcommand{\MaxImprove}{\mathrm{Improve}_{\mathrm{Max}}}

\newtheorem{theorem}{Theorem}

\newtheorem{algorithm}{Algorithm}

\newtheorem{lemma}[theorem]{Lemma}

\newtheorem{proposition}[theorem]{Proposition}
\newtheorem{corollary}[theorem]{Corollary}

\newcommand{\set}[1]{\{ #1 \}}
\newcommand{\Set}[1]{\big\{ #1 \big\}}
\newcommand{\eset}[1]{\{ \: #1 \: \}}

\newcommand{\seq}[1]{\langle #1 \rangle}
\newcommand{\Rplus}{{\mathbb R}_{\geq 0}} 
\newcommand{\Npos}{\mathbb N_{>0}}
\newcommand{\Nat}{\mathbb N}
\newcommand{\Real}{\mathbb R}

\newcommand{\Int}{\mathbb{Z}}

\newcommand{\obciach}{{\upharpoonright}}

\newcommand{\out}[1]{}

\newcommand{\Aa}{{\cal A}}

\newcommand{\Mm}{{\cal M}}

\newcommand{\Rr}{{\cal R}}

\newcommand{\Tt}{{\cal T}}

\newcommand{\CC}{\mathrm{CC}}
\newcommand{\CHOOSE}{\mathrm{Choose}}
\newcommand{\RESET}{\mathrm{Reset}}

\newcommand{\FRUNS}{\mathrm{Runs}_{\mathrm{fin}}}
\newcommand{\LAST}{\mathrm{Last}}
\newcommand{\STOP}{\mathrm{Stop}}

\newcommand{\CLOS}[1]{\overline{#1}}
\newcommand{\LENGTH}{\mathrm{Length}}

\newcommand{\VAL}{\mathrm{Val}}

\newcommand{\SUCC}{\mathrm{Succ}}

\newcommand{\RUN}{\mathrm{Run}}
\newcommand{\FLOOR}[1]{\lfloor #1 \rfloor}
\newcommand{\TIME}{\mathrm{Time}}
\newcommand{\REACHTIME}{\mathrm{RT}}

\newcommand{\Opt}{\mathrm{Opt}}
\newcommand{\THIN}{\mathrm{Thin}}
\newcommand{\THICK}{\mathrm{Thick}}

\newcommand{\mMIN}{\mathrm{Min}}
\newcommand{\mMAX}{\mathrm{Max}}
\newcommand{\mMINMAX}{\mathrm{MinMax}}

\newcommand{\NATS}[1]{\llbracket #1 \rrbracket_\Nat}

\newcommand{\REALS}[1]{\llbracket #1 \rrbracket_\Real}

\newcommand{\FRAC}[1]{\lbag #1 \rbag}

\title{Reachability-time games on timed automata%
  \thanks{This research was supported in part by EPSRC project
    EP/E022030/1.} 
}
\author{Marcin Jurdzi{\'n}ski%
  \thanks{Part of this work was done when the author visited the Isaac
    Newton Institute for Mathematical Sciences, Cambridge.
    Financial support from the Institute is gratefully acknowledged.} 
  ~and Ashutosh Trivedi \\
  {\it Department of Computer Science}, 
  {\it  University of Warwick, UK} \\
  {\it \{mju,trivedi\}@dcs.warwick.ac.uk}
}
\date{}

\begin{document}

\maketitle

\begin{abstract}
  In a reachability-time game, players Min and Max choose moves so
  that the time to reach a final state in a timed automaton is
  minimised or maximised, respectively.   
  Asarin and Maler showed decidability of reachability-time games on
  strongly non-Zeno timed automata using a value iteration algorithm. 
  This paper complements their work by providing a strategy
  improvement algorithm for the problem. 
  It also generalizes their decidability result because the proposed
  strategy improvement algorithm solves reachability-time games on all
  timed automata. 
  The exact computational complexity of solving reachability-time
  games is also established:
  the problem is EXPTIME-complete for timed automata with at least two
  clocks. 
\end{abstract}

\section{Introduction}

Timed automata~\cite{AD94} are a fundamental formalism for modelling
and analysis of real-time systems.
They have rich theory, solid modelling and verification tool 
support~\cite{Uppaal,Hytech,Kronos}, and they have been successfully 
applied to numerous industrial case studies.
Timed automata are finite automata augmented by a finite number of
continuous real variables which are called clocks because their values
increase with time at unit rate.
Every clock can be reset to an integer constant when a transition of
the automaton is performed, and clock values can be compared to
integers to constrain availability of transitions. 
Adding clocks to finite automata 
increases their expressive power and the fundamental reachability  
problem is PSPACE-complete for timed automata~\cite{AD94}.
The natural optimization problems of minimizing and maximizing 
reachability-time in timed automata are also in PSPACE~\cite{CY92}. 

The reachability (or optimal reachability-time) problems in timed 
automata are fundamental to the \emph{verification} of 
(quantitative timing) properties of systems modeled by timed
automata~\cite{AD94}.  
On the other hand, the problem of \emph{control-program synthesis} for
real-time systems can be cast as a two-player reachability 
(or optimal reachability-time) games, where the two players, say Min
and Max, correspond to the ``controller'' and the ``environment'', 
respectively, and control-program synthesis corresponds to computing
winning (or optimal) strategies for Min.  
In other words, for control-program synthesis we need to generalize 
optimization problems to \emph{competitive optimization} problems. 
Reachability games~\cite{AMPS98} and reachability-time
games~\cite{AM99} on timed automata are decidable.
The former problem is EXPTIME-complete, but the elegant result of
Asarin and Maler~\cite{AM99} for reachability-time games is limited to
the class of strongly non-Zeno timed automata and no upper complexity
bounds are given.  
A recent result of Henzinger and Prabhu~\cite{HP06} is that values of
reachability-time games can be approximated for all timed automata,
but computatability of the exact values 
was left open. 

A generalization of timed automata to priced (or weighted) timed
automata~\cite{Bou06} allows a rich variety of applications,
e.g., to scheduling~\cite{BFHLPRV01,AAM06,RLS06,CORA}.
While the fundamental minimum reachability-price problem is 
PSPACE-complete~\cite{BFHLPRV01,BBBR07}, the two-player
reachability-price games are undecidable on priced timed automata with
at least three clocks~\cite{BBM06}. 
The reachability-price games are, however, decidable for priced timed 
automata with one clock~\cite{BLMR06}, and on the class of strongly
price-non-Zeno priced timed automata~\cite{ABM04,BCFL04}.  

\paragraph{Our contribution.}

We show that the exact values 
of reachability-time games on arbitrary timed automata are uniformly
computable;
here uniformity means that the output of our algorithm allows us, for
every starting state, to compute in constant time the value of the
game starting from this state.
In particular, unlike the paper of Asarin and Maler~\cite{AM99}, we do
not require timed automata to be strongly non-Zeno.  
We also establish the exact complexity of reachability-time 
games: they are EXPTIME-complete and two clocks are sufficient for
EXPTIME-hardness. 
For the latter result we reduce from a~recently discovered
EXPTIME-complete problem of countdown games~\cite{JLS07}.   

We believe that an important contribution of this paper are the novel
proof techniques used. 
We characterize the values of the game by \emph{optimality equations}
and then we use \emph{strategy improvement} to solve them.
This allows us to obtain an elementary and constructive proof of the 
fundamental determinacy result for reachability-time games, which at
the same time yields an efficient algorithm matching the EXPTIME lower
bound for the problem. 
Those techniques were known for finite state systems~\cite{Put94,VJ00}
but we are not aware of any earlier algorithmic results based on
optimality equations and strategy improvement for real-time systems
such as timed automata.

\paragraph{Related and future work.}

A recent, concurrent, and independent work~\cite{BHPR07}
establishes decidability of slightly different and more challenging 
reachability-time games 
``with the element of surprise''~\cite{dAFHMS03,HP06}. 
In our model of timed games players take turns to take unilateral
decisions about the duration and type of subsequent game moves. 
Games with surprise are more general in two ways: 
in every round of the game players have a ``time race'' to be the
first to perform a move;
moreover, players are forbidden to use strategies which 
``stop the time'', because such strategies are arguably physically
unrealistic and result in Zeno runs.

We conjecture that our principal technique of optimality equations and
strategy improvement can be generalized to give an EXPTIME algorithm 
for reachability-time games with surprise, and we are currently
working on it.
We also believe that this technique is applicable to many other 
(competitive) optimization problems on (priced) timed automata and
even on restricted classes of hybrid automata;
we are currently working on optimality equations and strategy
improvement for, e.g., average-time games on timed automata and on
o-minimal hybrid systems~\cite{BBC07}.

\section{Reachability-time games}
\label{section:reachability-time-games}

We assume
that, wherever appropriate, sets $\Nat$ of non-negative
integers and $\Real$ of reals contain a maximum element $\infty$, and
we write 
$\Npos$ for the set of positive integers and
$\Rplus$ for the set of non-negative reals. 
For $n \in \Nat$, we write $\NATS{n}$ for the set 
$\set{0, 1, \dots, n}$, 
and $\REALS{n}$ for the set 
$\set{r \in \Real \: : \: 0 \leq r \leq n}$ of non-negative reals 
bounded by~$n$.
For $r \in \Rplus$, we write $\FLOOR{r}$ for its integer part, 
and we write $\FRAC{r}$ for its fractional part.
For sets $X$ and $Y$, we write $[X \to Y]$ for the set of functions
$F : X \to Y$, and $[X \rightharpoondown Y]$ for the set of partial
functions $F : X \rightharpoondown Y$.

\paragraph{Timed automata.}

Fix a constant $k \in \Nat$ for the rest of this paper.
Let $C$ be a finite set of \emph{clocks}. 
A ($k$-bounded) \emph{clock valuation} is a function 
$\nu : C \to \REALS{k}$;
we write $V$ for the set $[C \to \REALS{k}]$ of clock valuations. 
If~$\nu \in V$ and $t \in \Rplus$ then we write $\nu + t$ for the
clock valuation defined by $(\nu + t)(c) = \nu(c) + t$, for all 
$c \in C$. 
For a set $C' \subseteq C$ of clocks and a clock valuation 
$\nu : C \to \Rplus$, we define $\RESET(\nu, C')(c) = 0$ if 
$c \in C'$, and $\RESET(\nu, C')(c) = \nu(c)$ if $c \not\in C'$.

The set of \emph{clock constraints} over the set of clocks $C$ is the 
set of conjunctions of \emph{simple clock constraints}, which are
constraints of the form $c \bowtie i$ or $c - c' \bowtie i$, where 
$c, c' \in C$, $i \in \NATS{k}$, and 
${\bowtie} \in \eset{<, >, =, \leq, \geq}$. 
Note that there are finitely many simple clock constraints and hence
the set of non-equivalent clock constraints is finite. 
For every clock valuation $\nu \in V$, let $\CC(s)$ be the set of
simple clock constraints which hold in~$\nu \in V$.  
A \emph{clock region} is a maximal set $P \subseteq V$, such that for
all $\nu, \nu' \in P$, we have $\CC(\nu) = \CC(\nu')$.
In other words, clock regions are equivalence classes of the
equivalence relation relating clock valuations which are
indistinguishable by clock constraints. 
Observe that 
$\nu$ and~$\nu'$ are in the same
clock region iff all clocks have the same integer parts in $\nu$
and~$\nu'$, and if the partial orders of the clocks determined by
their fractional parts in $\nu$ and $\nu'$ are the same. 
For all $\nu \in V$, we write $[\nu]$ for the clock region of $\nu$.

A \emph{clock zone} is a convex set of clock valuations which 
is a union of a set of clock regions. 
Note that a set of clock valuations is a zone iff it is definable by a
clock constraint.
For $W \subseteq V$, we write $\CLOS{W}$ for the closure of the set
$W$, i.e., the smallest closed set in~$V$ which contains~$W$. 
Observe that for every clock zone~$W$, the set $\CLOS{W}$ is also a
clock zone.  

Let $L$ be a finite set of \emph{locations}.
A \emph{configuration} is a pair $(\ell, \nu)$, where $\ell \in L$ is
a location and $\nu \in V$ is a clock valuation; 
we write~$Q$ for the set of configurations. 
If $s = (\ell, \nu) \in Q$ and $c \in C$, then we write $s(c)$ for
$\nu(c)$. 
A~\emph{region} is a pair $(\ell, P)$, where $\ell \in L$ is a
location and $P$ is a clock region.  
If $s = (\ell, \nu)$ is a configuration then we write $[s]$ for the
region $(\ell, [\nu])$.
We write~$\Rr$ for the set of regions.
A set $Z \subseteq S$ is a \emph{zone} if for every $\ell \in L$,
there is a clock zone $W_\ell$, such that 
$Z = \set{(\ell, \nu) \: : \: 
  \ell \in L \text{ and } \nu \in W_\ell}$.  
For a region $R = (\ell, P) \in \Rr$, we write $\CLOS{R}$ for the zone
$\set{(\ell, \nu) \: : \: \nu \in \CLOS{P}}$. 

  A \emph{timed automaton} $\Tt = (L, C, S, A, E, \delta, \rho, F)$ 
  consists of 
  a finite set of locations~$L$, 
  a finite set of clocks~$C$, 
  a set of \emph{states} $S \subseteq Q$, 
  a finite set of \emph{actions} $A$, 
  an \emph{action enabledness function} $E : A \to 2^S$,  
  a \emph{transition function} $\delta : L \times A \to L$, 
  a \emph{clock reset function} $\rho : A \to 2^C$, 
  and a set of \emph{final states} 
  $F \subseteq S$. 
  We futher require that $S$, $F$, and $E(a)$ for all $a \in A$, are
  zones. 

For a configuration $s = (\ell, \nu) \in Q$ and $t \in \Rplus$, we
define $s + t$ to be the configuration $s' = (\ell, \nu + t)$ if
$\nu+t \in V$, and we then write $s \xrightharpoonup{}_t s'$. 
We write $s \xrightarrow{}_t s'$ if $s \xrightharpoonup{}_t s'$ and 
for all $t' \in [0, t]$, we have $(\ell, s + t') \in S$. 
For an action $a \in A$, we define $\SUCC(s, a)$ to be the
configuration $s' = (\ell', \nu')$, where $\ell' = \delta(\ell, a)$
and $\nu' = \RESET(\nu, \rho(a))$, and we then write 
$s \xrightharpoonup{a} s'$.
We write $s \xrightarrow{a} s'$ if $s \xrightharpoonup{a} s'$; 
$s, s' \in S$; and $s \in E(a)$.
For technical convenience and without loss of generality we will 
assume throughout that timed automata satisfy the requirement that 
for every $s \in S$, there exists $a \in A$, such that 
$s \xrightarrow{a} s'$.

For $s, s' \in S$, we say that $s'$ is in the future of $s$, 
or equivalently, that $s$ is in the past of $s'$, 
if there is $t \in \Rplus$, such that $s \xrightarrow{}_t s'$; 
we then write $s \xrightarrow{}_* s'$. 
%
%
For $R, R' \in \Rr$, we say that $R'$ is in the future of $R$, 
or 
that $R$ is in the past of $R'$, if there is $s \in R$ and there is
$s' \in R'$, such that $s'$ is in the future of $s$; 
we then write $R \xrightarrow{}_* R'$.
We say that $R'$ is the \emph{time successor} of $R$ if 
$R \xrightarrow{}_* R'$, $R \not= R'$, and for every $R'' \in \Rr$,
we have that $R \xrightarrow{}_* R'' \xrightarrow{}_* R'$ implies 
$R'' = R$ or $R'' = R'$;
we then write $R \xrightarrow{}_{+1} R'$ or 
$R' \xleftarrow{}_{+1} R$. 
Similarly, for $R, R' \in \Rr$, we write $R \xrightarrow{a} R'$ if
there is 
$s \in R$, and there is $s' \in R'$, such that $s \xrightarrow{a} s'$.

We say that a region $R \in \Rr$ is \emph{thin} if for every $s \in R$
and every $\varepsilon > 0$, we have that 
$[s] \not= [s+\varepsilon]$;
other regions are called \emph{thick};
we write $\Rr_\THIN$ and $\Rr_\THICK$ for the sets of thin and thick
regions, respectively. 
Note that if $R \in \Rr_\THICK$ then for every $s \in R$, there is an
$\varepsilon > 0$, such that $[s] = [s + \varepsilon]$.  
Observe also, that the time successor of a thin region is thick and
vice versa. 

A \emph{timed action} is a pair $\tau = (a, t) \in A \times \Rplus$. 
For $s \in Q$, we define $\SUCC(s, \tau) = \SUCC(s, (a, t))$ to be the
configuration $s' = \SUCC(s + t, a)$, i.e., such that 
$s \xrightharpoonup{}_t s'' \xrightharpoonup{a} s'$, and we then write
$s \xrightharpoonup{a}_t s'$.  
We write $s \xrightarrow{a}_t s'$ if 
$s \xrightarrow{}_t s'' \xrightarrow{a} s'$.
If $\tau = (a, t)$ then we write $s \xrightharpoonup{\tau} s'$ instead
of $s \xrightharpoonup{a}_t s'$, and $s \xrightarrow{\tau} s'$ instead
of $s \xrightarrow{a}_t s'$. 

A finite run of a timed automaton is a sequence 
$\seq{s_0, \tau_1, s_1, \tau_2, \dots, \tau_{n}, s_n} \in S \times ((A
\times \Rplus) \times S)^*$, such that for all $i$, $1 \leq i \leq n$,
we have $s_{i-1} \xrightarrow{\tau_i} s_i$. 
For a finite run 
$r = \seq{s_0, \tau_1, s_1, \tau_2, \dots, \tau_{n}, s_n}$, 
we define $\LENGTH(r) = n$, and we define $\LAST(r) = s_n$ to be the
state in which the run ends. 
We write $\FRUNS$ for the set of finite runs. 
An infinite run of a timed automaton is a sequence 
$r = \seq{s_0, \tau_1, s_1, \tau_2, \dots}$, such that 
for all $i \geq 1$, we have $s_{i-1} \xrightarrow{\tau_i} s_i$. 
For an infinite run $r$, we define $\LENGTH(r) = \infty$. 
For a 
run $r = \seq{s_0, \tau_1, s_1, \tau_2, \dots}$, we define 
$\STOP(r) = \inf \set{i : s_i \in F}$ and 
$\TIME(r) = \sum_{i=1}^{\LENGTH(r)} t_i$; 
and we define $\REACHTIME(r) = \sum_{i=1}^{\STOP(r)} t_i$ if 
$\STOP(r) < \infty$, and $\REACHTIME(r) = \infty$ if 
$\STOP(R) = \infty$, 
where for all $i \geq 1$, we have $\tau_i = (a_i, t_i)$. 
%

\paragraph{Strategies.}

A reachability-time game $\Gamma$ is a triple 
$(\Tt, L_\mMIN, L_\mMAX)$, where 
$\Tt = (L, C, S, A, E, \delta, \rho, F)$ is a timed automaton and 
$(L_\mMIN, L_\mMAX)$ is a partition of 
$L$.
We define 
$Q_\mMIN = \set{(\ell, \nu) \in Q \: : \: \ell \in L_\mMIN}$,
$Q_\mMAX = Q \setminus Q_\mMIN$, 
$S_\mMIN = S \cap Q_\mMIN$, $S_\mMAX = S \setminus S_\mMIN$,
$\Rr_\mMIN = \set{[s] \: : \: s \in Q_\mMIN}$, and 
$\Rr_\mMAX = \Rr \setminus \Rr_\mMIN$. 

A \emph{strategy} for Min is a function 
$\mu : \FRUNS \to A \times \Rplus$, such that if 
$\LAST(r) = s \in S_\mMIN$ and $\mu(r) = \tau$ then 
$s \xrightarrow{\tau} s'$, where $s' = \SUCC(s, \tau)$. 
Similarly, a strategy for Max is a function 
$\chi : \FRUNS \to A \times \Rplus$, such that if 
$\LAST(r) = s \in S_\mMAX$ and $\chi(r) = \tau$ then 
$s \xrightarrow{\tau} s'$, where $s' = \SUCC(s, \tau)$. 
We write $\Sigma_\mMIN$ and $\Sigma_\mMAX$ for the sets of strategies
for Min and Max, respectively. 
If players Min and Max use strategies $\mu$ and $\chi$,
respectively, then the $(\mu, \chi)$-run from a state $s$ is the
unique run
$\RUN(s, \mu, \chi) = \seq{s_0, \tau_1, s_1, \tau_2, \dots}$,  
such that $s_0 = s$, and for every $i \geq 1$, if $s_i \in S_\mMIN$,
or $s_i \in S_\mMAX$, then 
$\mu(\RUN_i(s, \mu, \chi)) = \tau_{i+1}$, 
or $\chi(\RUN_i(s, \mu, \chi)) = \tau_{i+1}$, respectively, 
where 
$\RUN_i(s, \mu, \chi) = 
  \seq{s_0, \tau_1, s_1, \dots, s_{i-1}, \tau_i, s_i}$.

We say that a strategy $\mu$ for Min is \emph{positional} if
for all finite runs $r, r' \in \FRUNS$, we have that 
$\LAST(r) = \LAST(r')$ implies $\mu(r) = \mu(r')$.
A positional strategy for Min can be then represented as a
function $\mu : S_\mMIN \to A \times \Rplus$, which uniquely
determines the strategy $\mu^\infty \in \Sigma_\mMIN$ as follows:
$\mu^\infty(r) = \mu(\LAST(r))$, for all finite runs  
$r \in \FRUNS$. 
Positional strategies for Max are defined and represented in
the analogous way. 
We write $\Pi_\mMIN$ and $\Pi_\mMAX$ for the sets of positional
strategies for Min and for Max, respectively.

\paragraph{Value of reachability-time game and optimality equations 
  $\Opt(\Gamma)$.}
\label{subsection:value-optimality-equations}

For every $s \in S$, we define its \emph{upper value}
$\VAL^*(s)$ and its \emph{lower value} $\VAL_*(s)$ by
$\VAL^*(s) = \inf_{\mu \in \Sigma_\mMIN} 
  \sup_{\chi \in \Sigma_\mMAX} \REACHTIME(\RUN(s, \mu, \chi))$, 
and 
$\VAL_*(s) = \sup_{\chi \in \Sigma_\mMAX} 
  \inf_{\mu \in \Sigma_\mMIN} \REACHTIME(\RUN(s, \mu, \chi))$.
The inequality $\VAL_*(s) \leq \VAL^*(s)$ always holds.
A~reach\-ability-time game is \emph{determined} if for every
$s \in S$, its lower and upper values are equal to each other;  
then we say that the \emph{value} $\VAL(s)$ exists and 
$\VAL(s) = \VAL_*(s) = \VAL^*(s)$.
For strategies $\mu \in \Sigma_\mMIN$ and $\chi \in \Sigma_\mMAX$, we
define 
$\VAL^\mu(s) = \sup_{\chi \in \Sigma_\mMIN} 
  \REACHTIME(\RUN(s, \mu, \chi))$, and 
$\VAL_\chi(s) = \inf_{\mu \in \Sigma_\mMIN} 
  \REACHTIME(\RUN(s, \mu, \chi))$. 
For an $\varepsilon > 0$, we say that a strategy 
$\mu \in \Sigma_\mMIN$ or $\chi \in \Sigma_\mMAX$ is 
\emph{$\varepsilon$-optimal} if for every $s \in S$, we have
$\VAL^\mu(s) \leq \VAL(s) + \varepsilon$ or 
$\VAL_\chi(s) \geq \VAL(s) - \varepsilon$, respectively. 
Note that if a game is determined then for every $\varepsilon > 0$, 
both players have $\varepsilon$-optimal strategies. 

We say that a reachability-time game is \emph{positionally determined}
if for every~$s \in S$, we have  
$\VAL(s) = \inf_{\mu \in \Pi_\mMIN} \sup_{\chi \in \Sigma_\mMAX} 
  \REACHTIME(\RUN(s, \mu, \chi))$
and 
$\VAL(s) = \sup_{\chi \in \Pi_\mMAX} \inf_{\mu \in \Sigma_\mMIN} 
  \REACHTIME(\RUN(s, \mu, \chi))$.
Note that if the reachability-time game is positionally determined
then for every $\varepsilon > 0$, both players have \emph{positional} 
$\varepsilon$-optimal strategies. 
Our results (Lemma~\ref{lemma:opt-strageties-from-opt-eqn},
Theorem~\ref{theorem:correctness-of-reduction}, and
Theorem~\ref{theorem:2-player-regionally-simple-solution})
yield a constructive proof of the following fundamental result for 
reachability-time games. 

\begin{theorem}[Positional determinacy]
  Reachability-time games are positionally determined.
\end{theorem}


Let $\Gamma$ be a reachability-time game, and let $T : S \to \Real$
and $D : S \to \Nat$. 
We write $(T, D) \models \Opt_\mMINMAX(\Gamma)$, and we say that 
$(T, D)$ is a solution of optimality equations
$\Opt_\mMINMAX(\Gamma)$, if for all $s \in S$, we have: 
\begin{bzitemize}
\item
  if $D(s) = \infty$ then $T(s) = \infty$; and
  if $s \in F$ then $(T(s), D(s)) = (0, 0)$;
\item
  if $s \in S_\mMIN \setminus F$ then 
  $T(s) = \inf_{a, t} \set{t + T(s') \: : \: s \xrightarrow{a}_t s'}$,
  and 
$D(s) = \min \Set{1 + d' \: : \: 
  T(s) = \inf_{a, t} \set{t + T(s') \: : \: 
    s \xrightarrow{a}_t s' \text{ and } D(s') = d'}}$;
and
\item
  if $s \in S_\mMAX \setminus F$ then 
  $T(s) = \sup_{a, t} 
    \set{t + T(s') \: : \: s \xrightarrow{a}_t s'}$, 
  and
$
    D(s) = \max \Set{1 + d' \: : \: 
      T(s) = \sup_{a, t} \set{t + T(s') \: : \: 
        s \xrightarrow{a}_t s' \text{ and } D(s') = d'}}.
$
\end{bzitemize}

\begin{lemma}[$\varepsilon$-Optimal strategies from optimality
  equations]
\label{lemma:opt-strageties-from-opt-eqn}
  If $(T, D) \models \Opt_\mMINMAX(\Gamma)$, then for all $s \in S$,
  we have $\VAL(s) = T(s)$ and for every $\varepsilon > 0$, both
  players have \emph{positional} $\varepsilon$-optimal strategies.   
\end{lemma}

\paragraph{Simple functions and simple timed actions.}
\label{subsection:simple-functions-and-actions}


  Let $X \subseteq Q$. 
  A function $F : X \to \Real$ is \emph{simple} if either:
    there is $e \in \Int$, such that for every $s \in X$, we have
    $F(s) = e$; or 
    there are $e \in \Int$ and $c \in C$, such that for every 
    $s \in X$, we have $F(s) = e - s(c)$.   

Let $X \subseteq Q$ be convex and let $F : X \to \Real$ be a
continuous function.
We write $\CLOS{F}$ for the unique continuous function 
$F' : \CLOS{X} \to \Real$, such that for all $s \in X$, we have
$F'(s) = F(s)$. 
Observe that if $F$ is simple, then $\CLOS{F}$ is simple.  
For functions $F, F' : X \to \Real$ we define functions 
$\max(F, F'), \min(F, F') : X \to \Real$ by
$\max(F, F')(s) = \max \eset{F(s), F'(s)}$ and 
$\min(F, F')(s) = \min \eset{F(s), F'(s)}$, for every $s \in X$. 

\begin{lemma}
\label{lemma:simple-min-max-simple}
  Let $F, F' : R \to \Real$ be simple functions defined on a region  
  $R \in \Rr$. 
  Then either $\min(\overline{F}, \overline{F'}) = \overline{F}$ and
  $\max(\overline{F}, \overline{F'}) = \overline{F'}$, or  
  $\min(\overline{F}, \overline{F'}) = \overline{F'}$ and 
  $\max(\overline{F}, \overline{F'}) = \overline{F}$. 
  In particular, both $\min(\overline{F}, \overline{F'})$ and 
  $\max(\overline{F}, \overline{F'})$ are simple functions. 
\end{lemma}


Define the finite set of \emph{simple timed actions} 
$\Aa = A \times \NATS{k} \times C$.
For $s \in Q$ and $\alpha = (a, b, c) \in \Aa$, we
define $t(s, \alpha) = b - s(c)$ if $s(c) \leq b$, and 
$t(s, \alpha) = 0$ if $s(c) > b$; 
and we define $\SUCC(s, \alpha)$ to be the state 
$s' = \SUCC(s, \tau(\alpha))$, 
where $\tau(\alpha) = (a, t(s, \alpha))$;
we then write $s \xrightharpoonup{\alpha} s'$.
We also write $s \xrightarrow{\alpha} s'$ if 
$s \xrightarrow{\tau(\alpha)} s'$. 
Note that if $\alpha \in \Aa$ and $s \xrightarrow{\alpha} s'$ then
$[s'] \in \Rr_\THIN$. 
Observe 
that for every thin region $R' \in \Rr_\THIN$, there is a number 
$b \in \NATS{k}$ and a clock $c \in C$, such that for every 
$R \in \Rr$ in the past of $R'$, we have that $s \in R$ implies 
$(s + (b - s(c)) \in R'$;  
we then write $R \xrightarrow{}_{b, c} R'$.
For 
$\alpha = (a, b, c) \in \Aa$ and $R, R' \in \Rr$, we write 
$R \xrightarrow{\alpha} R'$ or $R \xrightarrow{a}_{b, c} R'$, if 
$R \xrightarrow{}_{b, c} R'' \xrightarrow{a} R'$, for some 
$R'' \in \Rr_\THIN$.    
For $\alpha \in \Aa$ and $R, R' \in \Rr$, if 
$R \xrightarrow{\alpha} R'$ and $F : R' \to \Real$ then we define the
functions $F^\oplus_\alpha : R \to \Real$ and 
$F^\boxplus_\alpha : R \to \Real$ by  
$F^\oplus_\alpha(s) = t(s, \alpha) + F(\SUCC(s, \alpha))$ and 
$F^\boxplus_\alpha(s) = 1 + F(\SUCC(s, \alpha))$, for all $s \in R$. 

\begin{proposition}
\label{proposition:t-alpha-simple}
  Let $\alpha \in \Aa$ and $R, R' \in \Rr$. 
  If $R \xrightarrow{\alpha} R'$ and $F : R' \to \Real$ is simple,
  then $F^\oplus_\alpha$ is simple.  
\end{proposition}


For $a \in A$ and $R, R', R'' \in \Rr$, if 
$R \xrightarrow{}_* R'' \xrightarrow{a} R'$, $s \in R$, and 
$F : R' \to \Real$, then we define the partial function 
$F^\oplus_{s, a} : \Rplus \rightharpoondown \Real$ by  
$F^\oplus_{s, a}(t) = t + F(\SUCC(s, (a, t)))$, 
for all $t \in \Rplus$, such that $(s+t) \in R''$;
note that the domain
$\set{t \in \Rplus \: : \: (s + t) \in R''}$ of $F^\oplus_{s, a}$ is
an interval.

\begin{proposition}
\label{proposition:nondecreasing}
  Let $a \in A$ and $R, R', R'' \in \Rr$. 
  If $R \xrightarrow{}_* R'' \xrightarrow{a} R'$, $s \in R$, 
  and $F : R' \to \Real$ is simple, then 
  $F^\oplus_{s, a}: I \to \Real$, where   
  $I = \set{t \in \Rplus \: : \: (s + t) \in R''}$, is continuous and
  nondecreasing.   
\end{proposition}


\section{Timed region graph}
\label{section:timed-region-graph}

\paragraph{Timed region graph $\widehat{\Gamma}$.} 

Let $\Gamma = (\Tt, L_\mMIN, L_\mMAX)$ be a reachability-time game.
We define the \emph{timed region graph} $\widehat{\Gamma}$ to be the
finite edge-labelled graph $(\Rr, \Mm)$, where the set $\Rr$ of
regions of timed automaton $\Tt$ is the set of vertices, and the
labelled edge relation $\Mm \subseteq \Rr \times \Aa \times \Rr$ is
defined in the following way.   
For $\alpha = (a, b, c) \in \Aa$ and $R, R' \in \Rr$ 
we have $(R, \alpha, R') \in \Mm$, sometimes denoted by  
$R \overset{\alpha}{\leadsto} R'$, if and only if one of the following
conditions holds:
\begin{bzitemize}
\item
  there is an $R'' \in \Rr$, such that  
  $R \xrightarrow{}_{b, c} R'' \xrightarrow{a} R'$; or
\item
  $R \in \Rr_{\mMIN}$, and there are $R'', R''' \in \Rr$, 
  such that $R \xrightarrow{}_{b, c} R'' \xrightarrow{}_{+1}
    R''' \xrightarrow{a} R'$; 
  or  
\item
  $R \in \Rr_{\mMAX}$, and there are $R'', R''' \in \Rr$, 
  such that 
  $R \xrightarrow{}_{b, c} R'' \xleftarrow{}_{+1} 
    R''' \xrightarrow{a} R'$. 
\end{bzitemize}
Observe that in all the cases above we have that $R'' \in \Rr_\THIN$
and $R''' \in \Rr_\THICK$. 
The motivation for the second case is the following. 
Let $R \to_* R''' \xrightarrow{a} R'$, where $R \in \Rr_\mMIN$ and
$R''' \in \Rr_\THICK$.  
One of the main results that we will implicitly establish is that in a
state $s \in R$, among all $t \in \Rplus$, such that $s+t \in R'''$,
the smaller the~$t$, the ``better'' the timed action $(a, t)$ is for
player Min.  
Note, however, that the set $\set{t \in \Rplus \: : \: s+t \in R'''}$
is an open interval because $R''' \in \Rr_\THICK$, and hence it does
not have the smallest element.
Therefore, for every $s \in R$, we model the ``best'' time to
wait, when starting from $s$, before performing an $a$-labelled
transition from region $R'''$ to region $R'$, by taking the infimum of
the set $\set{t \in \Rplus \: : \: s+t \in R'''}$.   
Observe that this infimum is equal to the $t_{R''} \in \Rplus$, such
that $s+t_{R''} \in R''$, where $R'' \xrightarrow{}_{+1} R'''$, and
that $t_{R''} = b - s(c)$, where $R \to_{b, c} R''$.
In the timed region graph $\widehat{\Gamma}$, we summarize this model
of the ``best'' timed action from region $R$ to region $R'$ via region
$R'''$, by having a move $(R, \alpha, R') \in \Mm$, where 
$\alpha = (a, b, c)$. 
The motivation for the first and the third cases of the definition of
$\Mm$ is similar.

\paragraph{Regional functions and optimality equations
  $\Opt_\mMINMAX(\widehat{\Gamma})$.} 

Recall from Section~\ref{subsection:value-optimality-equations} that a
solution of optimality equations $\Opt_\mMINMAX(\Gamma)$ for a
reachability-time game $\Gamma$ is a pair of functions $(T, D)$, such
that $T : S \to \Real$ and $D : S \to \Nat$.  
Our goal is to define analogous optimality equations
$\Opt_\mMINMAX(\widehat{\Gamma})$ for the timed region
graph~$\widehat{\Gamma}$. 

If $R \overset{\alpha}{\leadsto} R'$, where $R, R' \in \Rr$ and
$\alpha \in \Aa$, then $s \in R$ does not in general imply that
$\SUCC(s, \alpha) \in R'$;  
it is however the case that $s \in R$ implies 
$\SUCC(s, \alpha) \in \CLOS{R'}$. 
In order to correctly capture the constraints for successor states 
which fall out of the ``target'' region $R'$ of a move of the form  
$R \overset{\alpha}{\leadsto} R'$, we consider, as solutions of
optimality equations $\Opt_\mMINMAX(\widehat{\Gamma})$, 
\emph{regional functions} of types
$T : \Rr \to [S \rightharpoondown \Real]$ and  
$D : \Rr \to [S \rightharpoondown \Nat]$, where for every $R \in \Rr$,
the domain of partial functions $T(R)$ and $D(R)$ is $\CLOS{R}$. 
Sometimes, when defining a regional function 
$F : \Rr \to [S \rightharpoondown \Real]$, it will only be natural to
define $F(R)$ for all $s \in R$, instead of all $s \in \CLOS{R}$.
This is not a problem, however, because as discussed in
Section~\ref{subsection:simple-functions-and-actions} defining
$F(R)$ on the region~$R$ uniquely determines the continuous extension
of $F(R)$ to~$\CLOS{R}$. 
For a function $F : \Rr \to [S \rightharpoondown \Real]$, we define
the function $\widetilde{F} : S \to \Real$ by 
$\widetilde{F}(s) = F([s])(s)$.


Let $T : \Rr \to [S \to \Real]$ and let $D : \Rr \to [S \to \Nat]$.
We write $(T, D) \models \Opt_\mMINMAX(\widehat{\Gamma})$ if for all
$s \in S$, we have the following: 
\begin{bzitemize}
\item
  if $s \in F$ then 
  $\big(\widetilde{T}(s), \widetilde{D}(s)\big) = (0, 0)$; 
\item
  if $s \in S_\mMIN$ then 
  $\big(\widetilde{T}(s), \widetilde{D}(s)\big) 
    =
    \minlex_{m \in \Mm}
      \Set{\big(T(R')^\oplus_\alpha(s), D(R')^\boxplus_\alpha(s)
        \big) \: : \: 
          m = ([s], \alpha, R')}$;
\item
  if $s \in S_\mMAX$ then 
  $\big(\widetilde{T}(s), \widetilde{D}(s)\big) 
    =
    \maxlex_{m \in \Mm}
      \Set{\big(T(R')^\oplus_\alpha(s), D(R')^\boxplus_\alpha(s)
        \big) \: : \: 
          m = ([s], \alpha, R')}$.
\end{bzitemize}

\paragraph{Solutions of $\Opt_\mMINMAX(\Gamma)$ from solutions of 
  $\Opt_\mMINMAX(\widehat{\Gamma})$.} 

In this subsection we show that the function 
$(T, D) \mapsto (\widetilde{T}, \widetilde{D})$ 
translates solutions of reachability-time optimality equations 
$\Opt_\mMINMAX(\widehat{\Gamma})$ for the timed region graph
$\widehat{\Gamma}$ to solutions of optimality equations
$\Opt_\mMINMAX(\Gamma)$ for the reachability-time game $\Gamma$. 
In other words, we establish that the function 
$\Gamma \mapsto \widehat{\Gamma}$ is a reduction from the problem of
computing values in reachability-time games to the problem of solving 
optimality equations for timed region graphs. 
Then in Section~\ref{section:strategy-improvement} we give an
algorithm to solve optimality equations for
$\Opt_\mMINMAX(\widehat{\Gamma})$.  

We say that a function $F : \Rr \to [S \rightharpoondown \Real]$ is 
\emph{regionally simple} or \emph{regionally constant}, respectively,
if for every region $R \in \Rr$, the function 
$F(R) : \CLOS{R} \to \Real$ is simple or constant, respectively. 

\begin{theorem}[Correctness of reduction to timed region graphs] 
\label{theorem:correctness-of-reduction}
  If $(T, D) \models \Opt_\mMINMAX(\widehat{\Gamma})$, $T$ is
  regionally simple, and $D$ is regionally constant, then  
  $(\widetilde{T}, \widetilde{D}) \models \Opt_\mMINMAX(\Gamma)$. 
\end{theorem}

\begin{proof}
  We need to show that for every $s \in S_{\mMIN} \setminus F$,
  we have:
  (a)~$\widetilde{T}(s) = \inf_{a, t} 
    \set{t + \widetilde{T}(s') \: : \: s \xrightarrow{a}_t s'}$;
  and (b)~$\widetilde{D}(s) = \min_{d' \in \Nat} \Set{1 + d' \: : \:  
    \widetilde{T}(s) = \inf_{a, t} \set{t + \widetilde{T}(s') 
      \: : \: s \xrightarrow{a}_t s' \text{ and } 
        \widetilde{D}(s') = d'}}$.
  The proof of the corresponding equalities for states
  $s \in S_{\mMAX} \setminus F$ is similar and omitted.
  We prove the equality~(a) here.
  \begin{eqnarray*}
    \widetilde{T}(s) 
    & = & \min_{m \in \Mm}
      \Set{T(R')^\oplus_\alpha(s)
        \: : \: m = ([s], \alpha, R')}
    \\
    & = & \min \Big\{
      \min_{R'', a, R'}
        \Set{T(R')^\oplus_{s, a}(b - s(c)) \: : \: 
          [s] \xrightarrow{}_{b, c} R'' \xrightarrow{a} R'}, 
    \\
    & & \mbox{} \hspace{2.6em} \min_{R'', a, R'} 
        \Set{T(R')^\oplus_{s, a}(b - s(c)) \: : \: 
          [s] \xrightarrow{}_{b, c} R'' \xrightarrow{}_{+1} R'''
            \xrightarrow{a} R'}\Big\}
    \\
    & = & \min_{R'', a, R'} 
      \Set{\inf_t \set{T(R')^\oplus_{s, a}(t) \: : \: [s + t] = R''}
        \: : \: [s] \xrightarrow{}_* R'' \xrightarrow{a} R'}
    \\
    & = & \min_{R'', a, R'} 
      \Set{\inf_t \set{t + \widetilde{T}(\SUCC(s, (a, t))) \: : \:
          [s + t] = R''} 
        \: : \: [s] \xrightarrow{}_{*} R'' \xrightarrow{a} R'}
    \\
    & = & \inf_{a, t} 
        \set{t + \widetilde{T}(s') \: : \: s \xrightarrow{a}_t s'}
  \end{eqnarray*}  
  The first equality holds by the assumption that 
  $T \models \Opt_\mMINMAX(\widehat{\Gamma})$.
  The second equality holds by the definition of the move relation
  $\Mm$ of the timed graph $\widehat{\Gamma}$, and because if 
  $\alpha = (a, b, c)$ then 
  \[
    T(R')^\oplus_\alpha(s) 
    \: = \: 
    b - s(c) + T(R')(\SUCC(s, (a, b - s(c)))
    \: = \: 
    T(R')^\oplus_{s, a}(b - s(c)).
  \]
  For the third equality we invoke regional simplicity of $T$ 
  which by Proposition~\ref{proposition:nondecreasing} implies that
  the function $T(R')^\oplus_{s, a}$ is continuous and nondecreasing. 
  If either $[s] \xrightarrow{}_{b, c} R'' \xrightarrow{a} R'$, or 
  $[s] \xrightarrow{}_{b, c} R''' \xrightarrow{}_{+1} R''
    \xrightarrow{a} R'$,   
  then we have that $\inf \set{t \: : \: [s + t] = R''} = b - s(c)$,
  and hence 
  \[
    \inf_t \set{T(R')^\oplus_{s, a}(t) \: : \: [s + t] = R''} 
    \: = \: 
    T(R')^\oplus_{s, a}(b - s(c)),
  \]
  because $T(R')^\oplus_{s, a}$ is continuous and nondecreasing. 
  The fourth equality holds because $[s + t] = R''$ and 
  $R'' \xrightarrow{a} R'$ imply that $[\SUCC(s, (a, t))] = R'$, and  
  hence $T(R')(\SUCC(s, (a, t))) = \widetilde{T}(\SUCC(s, (a, t)))$.
\end{proof}

\section{Solving optimality equations by strategy improvement}
\label{section:strategy-improvement}

\paragraph{Positional strategies.}

A positional strategy for player Max in a timed region graph
$\widehat{\Gamma}$ is a function $\chi : S_\mMAX \to \Mm$, such that
for every $s \in S_\mMAX$, we have $\chi(s) = ([s], \alpha, R)$,
for some $\alpha \in \Aa$ and $R \in \Rr$. 
A~strategy $\chi : S_\mMAX \to \Mm$ is \emph{regionally constant} if
for all $s, s' \in S_\mMAX$, we have that $[s] = [s']$ implies
$\chi(s) = \chi(s')$; we can then write $\chi([s])$ for $\chi(s)$. 
Positional strategies for player Min are defined analogously. 
We write $\Delta_\mMAX$ and $\Delta_\mMIN$ for the sets of positional
strategies for players Max and Min, respectively.

If $\chi \in \Delta_\mMAX$ is regionally constant then we define the
strategy subgraph $\widehat{\Gamma} \obciach \chi$ to be the subgraph
$(\Rr, \Mm_\chi)$ 
where $\Mm_\chi \subseteq \Mm$ consists of: all moves 
$(R, \alpha, R') \in \Mm$, such that $R \in \Rr_\mMIN$; and of all
moves $m = (R, \alpha, R')$, such that $R \in \Rr_\mMAX$ and 
$\chi(R) = m$. 
The strategy subgraph $\widehat{\Gamma} \obciach \mu$ for a regionally
constant positional strategy $\mu \in \Delta_\mMIN$ for player Min is
defined analogously. 
We say that $R \in \Rr$ is \emph{choiceless} in a timed region graph
$\widehat{\Gamma}$ if $R$ has a unique successor in
$\widehat{\Gamma}$.  
We say that $\widehat{\Gamma}$ is 0-player 
if all $R \in \Rr$ are choiceless in $\widehat{\Gamma}$; 
we say that $\widehat{\Gamma}$ is 1-player 
if either all $R \in \Rr_\mMIN$ or all $R \in \Rr_\mMAX$ are
choiceless in $\widehat{\Gamma}$;  
every timed region graph $\widehat{\Gamma}$ is 2-player. 
Note that if $\chi$ and $\mu$ are positional strategies in
$\widehat{\Gamma}$ for players Max and Min, respectively, then 
$\widehat{\Gamma} \obciach \chi$ and 
$\widehat{\Gamma} \obciach \mu$ are 1-player and 
$(\widehat{\Gamma} \obciach \chi) \obciach \mu$ is 0-player.

For functions $T : \Rr \to [S \to \Real]$ and
$D : \Rr \to [S \to \Real]$, and $s \in S_\mMAX$, we define sets
$M^*(s, (T, D))$ and $M_*(s, (T, D))$, respectively, of moves enabled
in $s$ which are (lexicographically) $(T, D)$-optimal for player Max
and Min, respectively: 
\begin{eqnarray*}
  M^*(s, (T, D)) & = & \argmaxlex_{m \in \Mm}
  \Set{\big( T(R')^\oplus_\alpha(s), D(R')^\boxplus_\alpha(s) \big) 
    \: : \: m = ([s], \alpha, R')}, \text{ and }
  \\
  M_*(s, (T, D)) & = & \argminlex_{m \in \Mm}
  \Set{\big( T(R')^\oplus_\alpha(s), D(R')^\boxplus_\alpha(s) \big) 
    \: : \: m = ([s], \alpha, R')}.
\end{eqnarray*}
Let $\CHOOSE : 2^\Mm \to \Mm$ be a function such that for every
non-empty set of moves $M \subseteq \Mm$, 
we have $\CHOOSE(M) \in M$.
For regional functions $T : \Rr \to [S \rightharpoondown \Real]$ and  
$D : \Rr \to [S \rightharpoondown \Nat]$, the canonical 
$(T, D)$-optimal strategies $\chi_{(T, D)}$ and $\mu_{(T, D)}$ for
player Max and Min, respectively, are defined by:
$\chi_{(T, D)}(s) = \CHOOSE(M^*(s, (T, D)))$, for every 
$s \in S_\mMAX$; and $\mu_{(T, D)}(s) = \CHOOSE(M_*(s, (T, D)))$, 
for every $s \in S_\mMIN$.

\paragraph{Optimality equations $\Opt(\widehat{\Gamma})$,
  $\Opt_\mMAX(\widehat{\Gamma})$, $\Opt_\mMIN(\widehat{\Gamma})$,
  $\Opt_\geq(\widehat{\Gamma})$ and $\Opt_\leq(\widehat{\Gamma})$.} 


Let $T : \Rr \to [S \to \Real]$ and $D : \Rr \to [S \to \Nat]$. 
We write $(T, D) \models \Opt_\mMAX(\widehat{\Gamma})$ or 
$(T, D) \models \Opt_\mMIN(\widehat{\Gamma})$, respectively, if for
all $s \in F$, we have 
$\big(\widetilde{T}(s), \widetilde{D}(s)\big) = (0, 0)$, and for all
$s \in S \setminus F$, we have, respectively:
\begin{eqnarray*}
  \big(\widetilde{T}(s), \widetilde{D}(s)\big) 
  & = & 
  \maxlex_{m \in \Mm}
    \Set{\big( T(R')^\oplus_\alpha(s), D(R')^\boxplus_\alpha(s) \big)
      \: : \: m = ([s], \alpha, R')},
  \text{ or} \\
  \big(\widetilde{T}(s), \widetilde{D}(s)\big) 
  & = &
  \minlex_{m \in \Mm}
    \Set{\big( T(R')^\oplus_\alpha(s), D(R')^\boxplus_\alpha(s) \big)
      \: : \: m = ([s], \alpha, R')}.
\end{eqnarray*}
If $\widehat{\Gamma}$ is 0-player then $\Opt_\mMAX(\widehat{\Gamma})$
and $\Opt_\mMIN(\widehat{\Gamma})$ are equivalent to each other and
denoted by $\Opt(\widehat{\Gamma})$.

We write $(T, D) \models \Opt_\geq(\widehat{\Gamma})$ or
$(T, D) \models \Opt_\leq(\widehat{\Gamma})$, resp., if for all
$s \in F$, we have 
$\big(\widetilde{T}(s), \widetilde{D}(s)\big)  
  \geq^{\mathrm{lex}} (0, 0)$ or 
$\big(\widetilde{T}(s), \widetilde{D}(s)\big)  
  \leq^{\mathrm{lex}} (0, 0)$, respectively;
and for all $s \in S \setminus F$, we have, respectively:
\begin{eqnarray*}
  \big(\widetilde{T}(s), \widetilde{D}(s)\big) 
  & \geq^{\mathrm{lex}} & 
  \maxlex_{m \in \Mm}
    \Set{\big( T(R')^\oplus_\alpha(s), D(R')^\boxplus_\alpha(s) \big)
      \: : \: m = ([s], \alpha, R')},
  \text{ or}
  \\
  \big(\widetilde{T}(s), \widetilde{D}(s)\big) 
  & \leq^{\mathrm{lex}} & 
  \minlex_{m \in \Mm}
    \Set{\big( T(R')^\oplus_\alpha(s), D(R')^\boxplus_\alpha(s) \big)
      \: : \: m = ([s], \alpha, R')}.
\end{eqnarray*}

\begin{proposition}[Relaxations of optimality equations]
\label{proposition:relaxation}
  If $(T, D) \models \Opt_\mMAX(\widehat{\Gamma})$ then 
  $(T, D) \models \Opt_\geq(\widehat{\Gamma})$, and 
  if $(T, D) \models \Opt_\mMIN(\widehat{\Gamma})$ then 
  $(T, D) \models \Opt_\leq(\widehat{\Gamma})$. 
\end{proposition}



\begin{lemma}[Solution of $\Opt(\widehat{\Gamma})$ is regionally
  simple] 
\label{lemma:solution-opt-regionally-simple}
  Let $\widehat{\Gamma}$ be a 0-player timed region graph.
  If $(T, D) \models \Opt(\widehat{\Gamma})$ then $T$ is regionally
  simple and $D$ is regionally constant. 
\end{lemma}

\paragraph{Solving 1-player maximum reachability-time optimality
  equations $\Opt_\mMAX(\widehat{\Gamma})$.} 
\label{subsection:1-player-si}

In this section we give a~strategy improvement algorithm for solving 
maximum reachability-time optimality equations
$\Opt_\mMAX(\widehat{\Gamma})$ for a~1-player timed region graph
$\widehat{\Gamma}$. 

We define the following strategy improvement operator $\MaxImprove$: 
\[
  \MaxImprove(\chi, (T, D))(s) =  
  \begin{cases}
    \chi(s) & 
      \text{if $\chi(s) \in M^*(s, (T, D))$}, 
    \\
    \CHOOSE(M^*(s, T)) & 
      \text{if $\chi(s) \not\in M^*(s, (T, D))$}.
  \end{cases}
\]
Note that $\MaxImprove(\chi, (T, D))(s)$ may differ from the canonical
$(T, D)$-optimal choice $\chi_{(T, D)}(s)$
only if $\chi(s)$ is itself $(T, D)$-optimal in state~$s$, i.e., if 
$\chi(s) \in M^*(s, (T, D))$.

\begin{lemma}[Improvement preserves regional constancy of strategies] 
\label{lemma:simple-max-improve-regional}
  If $\chi \in \Delta_\mMAX$ is regionally constant, 
  $T : \Rr \to [S \to \Real]$ is regionally simple, and 
  $D : \Rr \to [S \to \Nat]$ is regionally constant, then 
  $\MaxImprove(\chi, (T, D))$ is regionally constant.
\end{lemma}


\begin{algorithm}
\label{algorithm:1-player-si}
  {\bf Strategy improvement algorithm for
    $\Opt_{\mMAX}(\widehat{\Gamma})$.} 
  \begin{enumerate}
  \setlength{\itemsep}{-0.4ex}
  \item  
    (Initialisation)
    Choose a regionally constant positional strategy $\chi_0$ for
    player Max in $\widehat{\Gamma}$;
    set $i := 0$.
  \item 
    (Value computation) 
    Compute the solution $(T_i, D_i)$ of 
    $Opt(\widehat{\Gamma} \obciach {\chi_i})$. 
  \item 
    (Strategy improvement)
    If $\MaxImprove(\chi_i, (T_i, D_i)) = \chi_i$, then return 
    $(T_i, D_i)$. 
    \\
    Otherwise, set $\chi_{i+1} := \MaxImprove(\chi_i, (T_i, D_i))$;
    set $i := i + 1$; and goto step 2.  
  \end{enumerate}
\end{algorithm}

\begin{proposition}[Fixpoints of $\MaxImprove$ are solutions of
  $\Opt_\mMAX(\widehat{\Gamma})$]
\label{proposition:fixpoint-maximprove}
  Let $\chi \in \Delta_\mMAX$ and let 
  $(T^\chi, D^\chi) \models \Opt(\widehat{\Gamma} \obciach \chi)$. 
  If $\MaxImprove(\chi, (T^\chi, D^\chi)) = \chi$ then 
  $(T^\chi, D^\chi) \models \Opt_\mMAX(\widehat{\Gamma})$.
\end{proposition}

If $F, F' : \Rr \to [S \rightharpoondown \Real]$ then we write 
$F \leq F'$ if for all $R \in \Rr$, and for all $s \in \CLOS{R}$, we
have $F(R)(s) \leq F'(R)(s)$. 
Moreover, $F < F'$ if $F \leq F'$ and there is $R \in \Rr$ and 
$s \in R$, such that $F(R)(s) < F'(R)(s)$.   
If $F, G, F', G' : \Rr \to [S \rightharpoondown \Real]$ then 
$(F, G) \leq^{\mathrm{lex}} (F', G')$ if $F < F'$, or if $F = F'$ and
$G \leq G'$.

\begin{proposition}[Solution of $\Opt(\widehat{\Gamma})$ is the maximum solution of $\Opt_\leq(\widehat{\Gamma})$]
\label{proposition:opt-lp-solution-leq}
  Let $T, T_\leq : \Rr \to [S \to \Real]$ and 
  $D, D_\leq : \Rr \to [S \to \Nat]$ be such that
  $(T, D) \models \Opt(\widehat{\Gamma})$ and 
  $(T_\leq, D_\leq) \models \Opt_\leq(\widehat{\Gamma})$.  
  Then we have $(T_\leq, D_\leq) \leq^{\mathrm{lex}} (T, D)$, and if 
  $(T_\leq, D_\leq) \not\models \Opt(\widehat{\Gamma})$ then 
  $(T_\leq, D_\leq) <^{\mathrm{lex}} (T, D)$. 
\end{proposition}

\begin{proof}
  Our first goal is to establish that for every $s \in S$, we have 
  $(\widetilde{T_\leq}(s), \widetilde{D_\leq}(s)) 
  \leq_{\mathrm{lex}} 
  (\widetilde{T}(s), \widetilde{D}(s))$.   
  We proceed by induction on $\widetilde{D}(s)$, i.e., on the
  length of the $\chi_{(T, D)}$-path in 
  $\widehat{\Gamma}$ from $[s]$ to a final region. 
  The trivial base case is when $[s]$ is a final region, because then
  $(\widetilde{T}(s), \widetilde{D}(s)) = (0, 0)$ and
  $(\widetilde{T_\leq}(s), \widetilde{D_\leq}(s)) 
  \leq_{\mathrm{lex}} (0, 0)$.  
  Let $s \in S \setminus F$ be such that $\widetilde{D}(s) = n+1$. 
  Then $\widetilde{D}(\SUCC(s, \chi_{(T, D)}(s))) = n$ and if 
  $\chi_{(T, D)}(s) = ([s], \alpha, R')$ then we have the following:
  \begin{eqnarray}
    \label{eqnarray:t-leq}
    \big(\widetilde{T_\leq}(s), \widetilde{D_\leq}(s)\big) 
      \leq_{\mathrm{lex}} 
      \big(
        T_\leq(R')^\oplus_\alpha(s), D_\leq(R')^\boxplus_\alpha(s)
      \big) 
    \leq_{\mathrm{lex}} 
      \big(
        T(R')^\oplus_\alpha(s), D(R')^\boxplus_\alpha(s)
      \big)
    = \big(\widetilde{T}(s), \widetilde{D}(s)\big),
  \end{eqnarray}
  where the first inequality follows from 
  $(T_\leq, D_\leq) \models \Opt_\leq(\widehat{\Gamma})$, the second
  inequality follows from the induction hypothesis, and the last
  equality follows from $(T, D) \models \Opt(\widehat{\Gamma})$ and 
  $\chi_{(T, D)}(s) = ([s], \alpha, R')$. 
  This concludes the proof that  
  $(T_\leq, D_\leq) \leq_{\mathrm{lex}} (T, D)$.

  We prove that if 
  $(T_\leq, D_\leq) \not\models \Opt(\widehat{\Gamma})$ then there is
  $s \in S$, such that 
  $(\widetilde{T_\leq}(s), \widetilde{D_\leq}(s)) <_{\mathrm{lex}} 
  (\widetilde{T}(s), \widetilde{D}(s))$. 
  Indeed, if $(T_\leq, D_\leq) \not\models \Opt(\widehat{\Gamma})$
  then either 
  $(\widetilde{T_\leq}(s), \widetilde{D_\leq}(s)) 
    <_{\mathrm{lex}} (0, 0)$ 
  for some $s \in F$, or there is $s \in S \setminus F$, for which the
  first inequality in~(\ref{eqnarray:t-leq}) is strict and hence we
  get  
  $(\widetilde{T_\leq}(s), \widetilde{D_\leq}(s)) <_{\mathrm{lex}} 
  (\widetilde{T}(s), \widetilde{D}(s))$. 
\end{proof}

\begin{lemma}[Strict strategy improvement for Max]
\label{lemma:strict-improvement-max}
  Let $\chi, \chi' \in \Delta_\mMAX$,
  let $(T, D) \models \Opt_\mMIN(\widehat{\Gamma} \obciach \chi)$ and
  $(T', D') \models \Opt_\mMIN(\widehat{\Gamma} \obciach \chi')$, and
  let $\chi' = \MaxImprove(\chi, (T, D))$.
  Then $(T, D) \leq^{\mathrm{lex}} (T', D')$ and if $\chi \not= \chi'$
  then $(T, D) <^{\mathrm{lex}} (T', D')$. 
\end{lemma}


The following theorem is an immediate corollary of
Lemmas~\ref{lemma:solution-opt-regionally-simple}
and~\ref{lemma:simple-max-improve-regional} 
(the algorithm considers only regionally constant strategies),
of Lemma~\ref{lemma:strict-improvement-max} and finiteness of the
number of regionally constant positional strategies for Max
(the algorithm terminates), 
and of Proposition~\ref{proposition:fixpoint-maximprove} 
(the algorithm returns a solution of optimality equations).  

\begin{theorem}[Correctness and termination of strategy improvement 
  for $\Opt_\mMAX(\widehat{\Gamma})$] 
\label{theorem:1-player-regionally-simple-solution}
  The strategy improvement algorithm for
  $\Opt_\mMAX(\widehat{\Gamma})$ 
  terminates in finitely many steps and returns a solution $(T, D)$ of
  $\Opt_\mMAX(\widehat{\Gamma})$, 
  such that $T$ is regionally simple and $D$ is regionally constant. 
\end{theorem}

\paragraph{Solving 2-player reachability-time optimality equations 
  $\Opt_{\mMIN\mMAX}(\widehat{\Gamma})$.}  
\label{subsection:2-player-si}

In this section we give a strategy improvement algorithm for solving
optimality equations $\Opt_\mMINMAX(\widehat{\Gamma})$ for a 2-player
timed region graph~$\widehat{\Gamma}$.   
The structure of the algorithm is very similar to that of
Algorithm~\ref{algorithm:1-player-si}. 
The only difference is that in step~2.\ of every iteration we solve
1-player optimality equations 
$\Opt_\mMAX(\widehat{\Gamma} \obciach \mu)$ instead of 0-player
optimality equations $\Opt(\widehat{\Gamma} \obciach \chi)$.
Note that we can perform step~2.\ of
Algorithm~\ref{algorithm:2-player-si} below by using 
Algorithm~\ref{algorithm:1-player-si}.

We define the following strategy improvement operator $\MinImprove$: 
\[
  \MinImprove(\mu, (T, D))(s) = 
  \begin{cases}
    \mu(s) & 
      \text{if $\mu(s) \in M_*(s, (T, D))$}, 
    \\
    \CHOOSE(M_*(s, (T, D))) & 
      \text{if $\mu(s) \not\in M_*(s, (T, D))$}.
  \end{cases}
\]

\begin{lemma}[Improvement preserves regional constancy of strategies] 
\label{lemma:simple-min-improve-regional}
  If $\mu \in \Delta_\mMIN$ is regionally constant, 
  $T : \Rr \to [S \to \Real]$ is regionally simple, and 
  $D : \Rr \to [S \to \Real]$ is regionally constant, then
  $\MinImprove(\mu, (T, D))$ is regionally constant. 
\end{lemma}

\begin{algorithm}
\label{algorithm:2-player-si}
  {\bf Strategy improvement algorithm for solving
    $\Opt_\mMINMAX(\widehat{\Gamma})$.}  
  \begin{enumerate}
  \setlength{\itemsep}{-0.4ex}
  \item  
    (Initialisation)
    Choose a regionally constant positional strategy $\mu_0$ for
    player Min in $\widehat{\Gamma}$;
    set $i := 0$. 
  \item 
    (Value computation) 
    Compute the solution $(T_i, D_i)$ of 
    $Opt_\mMAX(\widehat{\Gamma} \obciach {\mu_i})$. 
  \item 
    (Strategy improvement)
    If $\MinImprove(\mu_i, (T_i, D_i)) = \mu_i$, then return 
    $(T_i, D_i)$.
    \\
    Otherwise, set $\mu_{i+1} := \MinImprove(\mu_i, (T_i, D_i))$;
    set $i := i + 1$; and goto step 2.  
  \end{enumerate}
\end{algorithm}

\begin{proposition}[Fixpoints of $\MinImprove$ are solutions of
\label{proposition:fixpoint-minimprove}
  $\Opt_\mMINMAX(\widehat{\Gamma})$] 
  Let $\mu \in \Delta_\mMIN$ and 
  $(T^\mu, D^\mu) \models \Opt_\mMAX(\widehat{\Gamma} \obciach \mu)$.
  If $\MinImprove(\mu, (T^\mu, D^\mu)) = \mu$ then  
  $(T^\mu, D^\mu) \models \Opt_\mMINMAX(\widehat{\Gamma})$.
\end{proposition}

\begin{proposition}[Solution of $\Opt_\mMAX(\widehat{\Gamma})$ is the
  minimum solution of $\Opt_\geq(\widehat{\Gamma})$]
\label{proposition:opt-lp-solution-geq}
  Let $T, T_\geq : \Rr \to [S \to \Real]$ and 
  $D, D_\geq : \Rr \to [S \to \Real]$ be such that 
  $(T, D) \models \Opt_\mMAX(\widehat{\Gamma})$ and 
  $(T_\geq, D_\geq) \models \Opt_\geq(\widehat{\Gamma})$. 
  Then $(T_\geq, D_\geq) \geq^{\mathrm{lex}} (T, D)$, and if
  $(T_\geq, D_\geq) \not\models \Opt_\mMAX(\widehat{\Gamma})$ then
  $(T_\geq, D_\geq) >^{\mathrm{lex}} (T, D)$. 
\end{proposition}

\begin{lemma}[Strict strategy improvement for Min]
\label{lemma:strict-improvement-min}
  Let $\mu, \mu' \in \Delta_\mMIN$, let 
  $(T, D) \models \Opt_\mMAX(\widehat{\Gamma} \obciach \mu)$ and 
  $(T', D') \models \Opt_\mMAX(\widehat{\Gamma} \obciach \mu')$, and
  let $\mu' = \MinImprove(\mu, (T, D))$.
  Then $(T, D) \geq^{\mathrm{lex}} (T', D')$ and if $\mu \not= \mu'$
  then $(T, D) >^{\mathrm{lex}} (T', D')$. 
\end{lemma}

\begin{proof}
  First we argue that 
  $(T, D) \models \Opt_\geq(\widehat{\Gamma} \obciach \mu')$ which by 
  Proposition~\ref{proposition:opt-lp-solution-geq} implies that 
  $(T, D) \geq^{\mathrm{lex}} (T', D')$.  
  Indeed for every $s \in S \setminus F$, if 
  $\mu(s) = ([s], \alpha, R)$ and $\mu'(s) = ([s], \alpha', R')$ then
  we have 
  \[
    \big(\widetilde{T}(s), \widetilde{D}(s)\big) 
    = 
    \big( T(R)^\oplus_\alpha(s), D(R)^\boxplus_\alpha(s) \big) 
    \\
    \geq^{\mathrm{lex}} 
      \big(
        T(R')^\oplus_{\alpha'}(s), D(R')^\boxplus_{\alpha'}(s)
      \big),
  \]
  where the equality follows from 
  $(T, D) \models \Opt_\mMAX(\widehat{\Gamma} \obciach \mu)$, and the
  inequality follows from the definition of $\MinImprove$. 
  Moreover, if $\mu \not= \mu'$ then there is 
  $s \in S_\mMIN \setminus F$ for which the above inequality is
  strict.  
  Then 
  $(T, D) \not\models \Opt_\mMAX(\widehat{\Gamma} \obciach \mu')$ 
  because every vertex $R \in \Rr_\mMIN$ in 
  $\widehat{\Gamma} \obciach \mu'$ has a unique successor, and hence
  again by Proposition~\ref{proposition:opt-lp-solution-geq} we
  conclude that $(T, D) >^{\mathrm{lex}} (T', D')$.    
\end{proof}

The following theorem is an immediate corollary of
Theorem~\ref{theorem:1-player-regionally-simple-solution}
and Lemma~\ref{lemma:simple-min-improve-regional}, 
of Lemma~\ref{lemma:strict-improvement-min} and finiteness of the
number of regionally constant positional strategies for Min,
and of Proposition~\ref{proposition:fixpoint-minimprove}.

\begin{theorem}[Correctness and termination of strategy improvement
  for $\Opt_\mMINMAX(\widehat{\Gamma})$] 
\label{theorem:2-player-regionally-simple-solution}
  The strategy improve\-ment algorithm for
  $\Opt_\mMINMAX(\widehat{\Gamma})$ 
  terminates in finitely many steps and returns a solution $(T, D)$ of
  $\Opt_\mMINMAX(\widehat{\Gamma})$, 
  such that $T$ is regionally simple and $D$ is regionally constant. 
\end{theorem}

\section{Complexity}
\label{section:complexity}

\begin{lemma}[Complexity of strategy improvement]
\label{lemma:si-complexity}
  Let $\widehat{\Gamma_0}$, $\widehat{\Gamma_1}$, and
  $\widehat{\Gamma_2}$ be 0-player, 1-player, and 2-player timed
  region graphs, respectively. 
  A solution of $\Opt(\widehat{\Gamma_0})$ can be computed in time
  $O(|\Rr|)$. 
  The strategy improvement algorithms for 
  $\Opt_\mMAX(\widehat{\Gamma_1})$ and
  $\Opt_\mMINMAX(\widehat{\Gamma_2})$ terminate in $O(|\Rr|)$
  iterations and hence run in $O(|\Rr|^2)$ and $O(|\Rr|^3)$ time,
  respectively.
\end{lemma}
Since the number $|\Rr|$ of regions is at most exponential in the size
of a timed automaton~\cite{AD94}, we conclude that the strategy
improvement algorithm solves reachability-time games in exponential  
time. 

\begin{corollary}
  The problem of solving reachability-time games is in EXPTIME. 
\end{corollary}

Courcoubetis and Yannakakis proved that the reachability problem for
timed automata with at least three clocks is
PSPACE-complete~\cite{CY92}. 
We complement their result by showing that solving 2-player
reachability games on timed automata with at least two clocks is
EXPTIME-complete.  
Note that the best currently known lower bound for the reachability
problem for timed automata with two clocks is
NP-hardness~\cite{LMS04}.  

\begin{theorem}[Complexity of reachability games on timed automata]
\label{theorem:reachability-games-exptime-complete}
  The problem of solving reachability games is EXPTIME-complete on
  timed automata with at least two clocks.
\end{theorem}

\begin{theorem}[Complexity of reachability-time games on timed
  automata] 
  The problem of solving reachability-time games is EXPTIME-complete 
  on timed automata with at least two clocks.
\end{theorem}


\bibliographystyle{latex8}
\bibliography{papers}


\section*{Appendix}

\subsection*{Proofs from
  Section~\ref{section:reachability-time-games}}

\begin{proof}[\bf Proof of
Lemma~\ref{lemma:opt-strageties-from-opt-eqn}
($\varepsilon$-Optimal strategies from optimality equations).]
  We show that for every $\varepsilon > 0$, there exists a positional
  strategy $\mu_\varepsilon : S_\mMIN \to A \times \Rplus$ for player
  Min, such that for every strategy $\chi$ for player Max, if 
  $s \in S$ is such that $D(s) < \infty$, then we have 
  $\REACHTIME(\RUN(s, \mu_\varepsilon, \chi)) \leq T(s) +
    \varepsilon$. 
  The proof, that for every $\varepsilon > 0$, there exists a 
  positional strategy $\chi_\varepsilon : S_\mMAX \to A \times \Rplus$
  for player Max, such that for every strategy $\mu$ for player Min,
  if $s \in S$ is such that $D(s) < \infty$ then we have 
  $\REACHTIME(\RUN(s, \mu, \chi_\varepsilon)) \geq T(s) - 
    \varepsilon$, 
  is similar and omitted. 
  The proof, that if $D(s) = \infty$ then player Max has a strategy to
  prevent ever reaching a final state, is routine and omitted as well.
  Together, these facts imply that $T$ is equal to the value function
  of the reachability-time game, and the positional strategies
  $\mu_\varepsilon$ and $\chi_\varepsilon$, defined in the proof below
  for all $\varepsilon > 0$, are $\varepsilon$-optimal. 

  For $\varepsilon' > 0$, $T : S \to \Real$, and 
  $s \in S_\mMIN \setminus F$, we say that a timed action 
  $(a, t) \in A \times \Rplus$ is $\varepsilon'$-optimal for $(T, D)$
  in $s$ if $s \xrightarrow{a}_t s'$, and
  \begin{eqnarray}
  \label{equation:min-d-decr}
    D(s') & \leq & D(s) - 1, \text{ and}
  \\
  \label{equation:min-epsilon-optimal-for-T}
    t + T(s') & \leq & T(s) + \varepsilon'. 
  \end{eqnarray}
  Observe that for every state $s \in S_\mMIN$ and for every
  $\varepsilon' > 0$, there is a $\varepsilon'$-optimal timed action
  for $(T, D)$ in~$s$ because $(T, D) \models \Opt_\mMINMAX(\Gamma)$. 
  Moreover, again by $(T, D) \models \Opt_\mMINMAX(\Gamma)$ we have
  that for every $s \in S_\mMAX \setminus F$ and timed action 
  $(a, t)$, such that $s \xrightarrow{a}_t s'$, we have 
  \begin{eqnarray}
  \label{equation:max-d-decr}
    D(s') & \leq & D(s) - 1, \text{ and}
  \\
  \label{equation:max-T}
    t + T(s') & \leq & T(s).
  \end{eqnarray}

  Let $\varepsilon > 0$;
  we define $\mu_\varepsilon : S_\mMIN \to A \times \Rplus$ by setting
  $\mu_\varepsilon(s)$, for every $s \in S_\mMIN$, to be a timed
  action which is $\varepsilon'(s)$-optimal for $(T, D)$ in~$s$, where
  $\varepsilon'(s) > 0$ is sufficiently small 
  (to be determined later).   
  Let~$\chi$ be an arbitrary strategy for player Max and let 
  $r = \RUN(s, \mu_\varepsilon, \chi) = \seq{s_0, (a_1, t_1), s_1,
    (a_2, t_2), \dots}$.
  Let $N = \STOP(r)$.
  Our goal is to prove that $\REACHTIME(r) \leq T(s) + \varepsilon$,
  i.e., that $T(s) \geq \sum_{k=1}^N t_k - \varepsilon$. 

  For every state $s \in S$, such that $D(s) < \infty$, define
  $\varepsilon'(s) = \varepsilon \cdot 2^{-D(s)}$.
  Note that if we add left- and right-hand sides of the
  inequalities~(\ref{equation:min-epsilon-optimal-for-T}) 
  or~(\ref{equation:max-T}), respectively, for all states $s_i$, and 
  $\varepsilon'(s_i)$-optimal timed actions $\mu_\varepsilon(s_i)$ if
  $s_i \in S_\mMIN$, where $i = 0, 1, \dots, N-1$, then we get 
  \[
    T(s) \: = \: T(s_0)
    \: \geq \: 
    \sum_{k=1}^{N} t_k - 
      \sum_{k=0}^{N-1} \varepsilon'(s_k) 
    \: \geq \: \sum_{k=0}^{N-1} t_k - \varepsilon.
  \]
  The first inequality holds by $T(s_N) = T(s_{\STOP(r)}) = 0$, 
  and the second inequality holds because
  \[
    \sum_{k=0}^{N-1} \varepsilon'(s_k) 
    \: = \: \sum_{k=0}^{N-1} (\varepsilon \cdot 2^{-D(s_k)})
    \: \leq \: \varepsilon \cdot \sum_{d=1}^{\infty} 2^{-d}
    \: \leq \: \varepsilon, 
  \]
  where the first inequality follows by~(\ref{equation:min-d-decr}) 
  and~(\ref{equation:max-d-decr}).

  It may be worth noting that if the finite values of the function $D$
  are bounded, i.e., if $B < \infty$, where 
  $B = \sup_{s \in S} \set{D(s) \: : \: D(s) < \infty}$, then in the
  above proof it is sufficient to define 
  $\varepsilon'(s) = \varepsilon / B$, for all $s \in S$, which gives
  arguably more realistically ``physically implementable''
  $\varepsilon$-optimal strategies. 
\end{proof}

\begin{proof}[\bf Proof of Lemma~\ref{lemma:simple-min-max-simple}
] 
  We prove the lemma for functions $\min(F, F')$ and $\max(F, F')$
  instead of $\min(\overline{F}, \overline{F'})$ and  
  $\max(\overline{F}, \overline{F'})$, respectively.
  Extending the result to the unique continuous extensions to
  $\overline{X}$ is routine.
  The case when both $F$ and $F'$ are constant functions is
  straightforward. 
  Hence it suffices to consider the following two cases.

  Case 1.
  Let $F(s) = e - s(c)$ and let $F'(s) = e'$, for some 
  $e, e' \in \Int$ and a clock $c \in C$.
  Note that for every state $s \in R$, we have 
  $\FLOOR{F'(s) - F(s)} = (e' - e) + \FLOOR{s(c)}$ and 
  hence $\FLOOR{F' - F}$ is a constant function in region~$R$.  
  Therefore either $F'(s) - F(s) \geq 0$ for all $s \in R$, or 
  $F'(s) - F(s) \leq 0$ for all $s \in R$, i.e., either 
  $\min(F, F') = F$ and $\max(F, F') = F'$, or  
  $\min(F, F') = F'$ and $\max(F, F') = F$.

  Case 2.
  Let $F(s) = e - s(c)$ and $F'(s) = e' - s(c')$, for some 
  $e, e' \in \Int$ and clocks $c, c' \in C$. 
  Note that for every state $s \in R$, we have
  $\FLOOR{F'(s) - F(s)} = (e' - e) + \FLOOR{s(c') - s(c)}$ and
  \[
    \FLOOR{s(c') - s(c)} = 
    \begin{cases}
      \FLOOR{s(c')} - \FLOOR{s(c)} & 
        \text{if $\FRAC{s(c')} \geq \FRAC{s(c)}$}, \\
      \FLOOR{s(c')} - \FLOOR{s(c)} - 1 & 
        \text{if $\FRAC{s(c')} < \FRAC{s(c)}$}.
    \end{cases}
  \]
  In particular, as in the previous case we have that 
  $\FLOOR{F' - F}$ is a constant function in region~$R$ and
  hence one of the functions $F$ or $F'$ is equal to $\max(F, F')$ 
  and the other is equal to $\min(F, F')$.     
\end{proof}

\begin{proof}[\bf Proof of Proposition~\ref{proposition:t-alpha-simple}]  
  Let $\alpha = (a, b, c)$. 
  If $F$ is a constant function, i.e., if there is some $e \in \Int$,
  such that for all $s' \in R'$, we have $F(s') = e$, then 
  $F^\oplus_\alpha(s) = t(s, \alpha) + e$. 
  If $s(c) > b$ for all $s \in R$, then $t(s, \alpha) = 0$ for all 
  $s \in R$, and hence $F^\oplus_\alpha(s) = e$ and $F^\oplus_\alpha$
  is simple.   
  If instead $s(c) \leq b$ for all $s \in R$, then 
  $F^\oplus_\alpha(s) = (b - s(c)) + e = (b + e) - s(c)$ and hence it
  is a simple function.

  The other case is when $F$ is not a constant function, i.e., if
  there are a constant $e \in \Int$ and a clock $c' \in C$, such that
  for all $s' \in R'$, we have $F(s') = e - s'(c')$. 
  We consider two subcases.

  If $c' \in \rho(a)$ then 
  $F^\oplus_\alpha(s) = t(s, a) + (e - s'(c')) = t(s, \alpha) + e$,
  because by the assumption that $c' \in \rho(a)$ we have that
  $s'(c') = 0$.
  If $s(c) > b$ for all $s \in R$, then $t(s, \alpha) = 0$ for all
  $s \in R$, and hence $F^\oplus_\alpha(s) = e$ which is a~simple 
  function. 
  If instead $s(c) \leq b$ for all $s \in R$, then 
  $F^\oplus_\alpha(s) = (b + e) - s(c)$ which is also a simple
  function.  

  If instead $c' \not\in \rho(a)$ then 
  $F^\oplus_\alpha(s) = t(s, \alpha) + (e - (s(c') + t(s, \alpha)))
    = e - s(c')$, 
  because by the assumption that $c' \not\in \rho(a)$ we have that
  $s'(c') = s(c') + t(s, \alpha)$, and hence $F^\oplus_\alpha$ is a
  simple function. 
\end{proof}

\begin{proof}[\bf Proof of Proposition~\ref{proposition:nondecreasing}]
  We consider two cases.
  If $F$ is a constant function, i.e., if there is $e \in \Int$, such
  that for all $s' \in R'$ we have $F(s') = e$, then 
  $F^\oplus_{s, a}(t) = t + F(\SUCC(s, (a, t))) = t + e$, which is a 
  continuous and nondecreasing function of $t$. 

  The other case is when $F$ is not a constant function, i.e., if
  there are a constant $e \in \Int$ and a clock $c' \in C$, such that
  for all $s' \in R'$, we have $F(s') = e - s'(c')$. 
  We consider two subcases.
  If $c' \in \rho(a)$ then $F^\oplus_{s, a}(t) = t + e$ which is
  continuous and nondecreasing. 
  If instead $c' \not\in \rho(a)$ then 
  $F^\oplus_{s, a}(t) = t + (e - (s+t)(c')) = t + e - (s(c') + t) =  
    e - s(c')$,
  i.e., $F^\oplus_{s, a}$ is a constant function and hence continuous
  and nondecreasing. 
\end{proof}

\subsection*{Proofs from Section~\ref{section:timed-region-graph}}

\begin{proof}[\bf Proof of
  Theorem~\ref{theorem:correctness-of-reduction} 
  (Correctness of reduction to timed region graphs).]
  Now we prove the equality~(b).
  \begin{eqnarray*}
    \widetilde{D}(s) 
    & = & \min_{m \in \Mm}
      \Set{D(R')^\boxplus_\alpha(s)
        \: : \:
        \widetilde{T}(s) = T(R')^\oplus_\alpha(s)
        \text{ and } 
        m = ([s], \alpha, R')}
    \\
    & = & \min_{d' \in \Nat}
      \Set{1 + d' \: : \:
        \widetilde{T}(s) = 
        T(R')^\oplus_\alpha(s)
        \text{ and } 
        ([s], \alpha, R') \in \Mm
        \text{ and }
        D(R') \equiv d'}
    \\
    & = & \min_{d' \in \Nat} \Set{1 + d' \: : \: 
      \widetilde{T}(s) = 
      \inf_{a, t} \set{t + \widetilde{T}(s') \: : \: 
        s \xrightarrow{a}_t s' \text{ and } \widetilde{D}(s') = d'}} 
  \end{eqnarray*}  
  The first equality holds by the assumption that 
  $(T, D) \models \Opt_\mMINMAX(\widehat{\Gamma})$.
  The second equality holds because of the assumption that $D$ is
  regionally constant, and we write $D(R') \equiv d'$, where 
  $d' \in \Nat$, to express that for all $s \in R'$, we have
  $D(R')(s) = d'$.
  Finally, to establish the third equality it is sufficient to perform
  a calculation analogous to the above proof
  of~(a), in order to show that
  \[
    \widetilde{T}(s) = 
    T(R')^\oplus_\alpha(s)
    \text{ and } 
    ([s], \alpha, R') \in \Mm
    \text{ and }
    D(R') \equiv d'
  \] 
  if and only if 
  \[
    \widetilde{T}(s) = 
    \inf_{a, t} \set{t + \widetilde{T}(s') \: : \: 
      s \xrightarrow{a}_t s' \text{ and } \widetilde{D}(s') = d'}. 
  \]
\end{proof}

\subsection*{Proofs from Section~\ref{section:strategy-improvement}}

\begin{proof}[\bf Proof of
Lemma~\ref{lemma:solution-opt-regionally-simple}
(Solution of $\Opt(\widehat{\Gamma})$ is regionally simple).] 
  In a 0-player timed region graph~$\widehat{\Gamma}$,
  for every region $R$, there is at most one outgoing labelled edge 
  $(R, \alpha, R') \in \Mm$, and hence for every region $R$, there is
  a~unique $\Mm$-path from $R$ in $\widehat{\Gamma}$. 
  For every region $R \in \Rr$, we define the distance
  $d(R) \in \Nat$ to be the smallest number of edges in the unique
  $\Mm$-path from $R$, that one needs to reach a final region. 
  It is easy to show that for every state $s \in S$, we have that
  $D([s])(s) = d([s])$, and hence $D$ is regionally constant.

  We prove that for every region $R \in \Rr$, the function 
  $T(R) : \CLOS{R} \to \Real$ is simple, by induction on $d(R)$.  
  If $d(R) = 0$ then $T(R)(s) = 0$ for all $s \in \CLOS{R}$, and hence
  $T(R)$ is simple on $\CLOS{R}$.

  Let $d(R) = n+1$ and let $(R, \alpha, R') \in \Mm$ be the unique
  edge going out of $R$ in $\widehat{\Gamma}$. 
  Observe that $T(R) = T(R')^\oplus_\alpha$ because for every 
  $s \in R$, we have $T(R)(s) = T([s])(s) = T(R')^\oplus_\alpha(s)$, 
  where the second equality follows from 
  $(T, D) \models \Opt(\widehat{\Gamma})$.  
  Moreover, by the induction hypothesis the function 
  $T(R') : \CLOS{R'} \to \Real$ is simple, and hence by
  Proposition~\ref{proposition:t-alpha-simple} we get that 
  $T(R')^\oplus_\alpha = T(R)$ is simple. 

  If $d(R) = \infty$, i.e., if the unique $\Mm$-path from $R$ in
  $\widehat{\Gamma}$ never reaches a final region, then we set
  $T(R')(s) = \infty$, for all $s \in \CLOS{R}$.
  Therefore $T(R') : \CLOS{R} \to \Real$ is a constant function and
  hence it is simple. 
\end{proof}

\begin{proof}[\bf Proof of
Lemma~\ref{lemma:simple-max-improve-regional} 
(Improvement preserves regional constancy of strategies)]
  We need to prove that for $s, s' \in S$, if $[s] = [s']$ then
  $\chi'(s) = \chi'(s')$, where $\chi' = \MaxImprove(\chi, (T, D))$. 
  By regionality of $\chi$ it is sufficient to prove that 
  $M^*(s, (T, D)) = M^*(s', (T, D))$.  
  By regional simplicity of $T$, and by
  Proposition~\ref{proposition:t-alpha-simple}, we have that functions
  $T(R)^\oplus_\alpha : [s] \to \Real$, for all 
  $m = ([s], \alpha, R) \in \Mm$, are simple.   
  Then we have
  \begin{eqnarray*}
    M^*(s, (T, D)) & = & \argmaxlex_{m \in \Mm}
      \Set{\big( T(R)^\oplus_\alpha(s), D(R)^\boxplus_\alpha(s) \big) 
        \: : \: m = ([s], \alpha, R)} \\ 
    & = & \argmaxlex_{m \in \Mm}
      \Set{\big(T(R)^\oplus_\alpha(s'), D(R)^\boxplus_\alpha(s')\big)
        \: : \: m = ([s'], \alpha, R)} \\ 
    & = & M^*(s', (T, D)),
  \end{eqnarray*}
  where the second equality follows from $[s] = [s']$, regional
  constancy of $D$, and by Lemma~\ref{lemma:simple-min-max-simple}
  applied to the (finite) set of functions  
  $\set{T(R)^\oplus_\alpha \: : \: ([s], \alpha, R) \in \Mm}$.  
\end{proof}

\begin{proof}[\bf Proof of Lemma~\ref{lemma:strict-improvement-max}
  (Strict strategy improvement for Max)]
  First we argue that 
  $(T, D) \models \Opt_\leq(\widehat{\Gamma} \obciach \chi')$ 
  which by Proposition~\ref{proposition:opt-lp-solution-leq} implies
  that $(T, D) \leq (T', D')$. 
  Indeed for every $s \in S \setminus F$, if 
  $\chi(s) = ([s], \alpha, R)$ and $\chi'(s) = ([s], \alpha', R')$
  then we have  
  \[
    \big(\widetilde{T}(s), \widetilde{D}(s)\big) 
    \: = \:
    \big(
      T(R)^\oplus_\alpha(s), D(R)^\boxplus_\alpha(s)
    \big) 
    \: \leq^{\mathrm{lex}} \:
      \big(
        T(R')^\oplus_{\alpha'}(s), D(R')^\boxplus_{\alpha'}(s) 
      \big),
  \]
  where the equality follows from 
  $(T, D) \models \Opt_\mMIN(\widehat{\Gamma} \obciach \chi)$, and the
  inequality follows from the definition of $\MaxImprove$. 
  Moreover, if $\chi \not= \chi'$ then there is 
  $s \in S_\mMAX \setminus F$ for which the above inequality is
  strict. 
  Then 
  $(T, D) \not\models \Opt_\mMIN(\widehat{\Gamma} \obciach \chi')$ 
  because every vertex in $\widehat{\Gamma} \obciach \chi'$ has a
  unique successor, and hence again by
  Proposition~\ref{proposition:opt-lp-solution-leq} we conclude that
  $(T, D) <^{\mathrm{lex}} (T', D')$.    
\end{proof}

\subsection*{Proofs from Section~\ref{section:complexity}}

\begin{proof}[\bf Proof of Lemma~\ref{lemma:si-complexity}
  (Complexity of strategy improvement)]
  An $O(|\Rr|)$ algorithm for solving $\Opt(\widehat{\Gamma_0})$ is
  implicit in the proof of
  Lemma~\ref{lemma:solution-opt-regionally-simple}.  

  Let $(T, D) \models \Opt_\mMAX(\widehat{\Gamma_1})$; and 
  for all $i \geq 0$, let $\chi_i \in \Delta_\mMAX$ be the strategy in
  the $i$-th iteration of Algorithm~\ref{algorithm:1-player-si}, and
  let 
  $(T_i, D_i) \models \Opt(\widehat{\Gamma_1} \obciach \chi_i)$. 
  We claim that for every $i \geq 0$, if $D(R) \equiv i$ then for all
  $j \geq i$, we have $(T_j(R), D_j(R)) = (T(R), D(R))$. 
  This can be established by a routine induction on the values of the 
  regionally constant 
  function~$D$.
  Observe that the finite values of the function $D$ are bounded
  by~$|\Rr|$, because in the proof of 
  Lemma~\ref{lemma:solution-opt-regionally-simple} they are set to be
  the length of a simple path in a timed region graph. 
  Algorithm~\ref{algorithm:1-player-si} must therefore terminate no
  later than after $|\Rr|+1$ iterations, because for every $i \geq 0$,
  in the $i$-th iteration there must be $R \in \Rr$ whose value $D(R)$
  is set to~$i$.  

  An analogous routine proof by induction on the value of $D$ can be
  used to prove that Algorithm~\ref{algorithm:2-player-si} terminates
  in $O(|\Rr|)$ iterations. 
\end{proof}


\begin{proof}[\bf Proof of
  Theorem~\ref{theorem:reachability-games-exptime-complete}
  (Complexity of reachability games on timed automata)]
  In order to solve a reachability game on a timed automaton it is
  sufficient to solve the reachability game on the finite region graph
  of the automaton.  
  Observe that every region, and hence also every configuration of
  the game, can be written down in polynomial space, and that every
  move of the game can be simulated in polynomial time.
  Therefore, the winner in the game can be determined by a
  straightforward alternating PSPACE algorithm, and hence the problem
  is in EXPTIME because APSPACE = EXPTIME.

  In order to prove EXPTIME-hardness of solving reachability games on
  timed automata with two clocks, we reduce the EXPTIME-complete
  problem of solving countdown games~\cite{JLS07} to it. 
  Let $G = (N, M, \pi, n_0, B_0)$ be a countdown game, where $N$ is a 
  finite set of nodes, $M \subseteq N \times N$ is a set of moves, 
  $\pi : M \to \Npos$ assigns a~positive integer number to every move,
  and $(n_0, B_0) \in N \times \Npos$ is the initial configuration. 
  In every move of the game from a configuration 
  $(n, B) \in N \times \Npos$, first player~1 chooses a number 
  $p \in \Npos$, such that $p \leq B$ and $\pi(n, n') = p$ for some
  move $(n, n') \in M$, and then player~2 chooses a move 
  $(n, n'') \in M$, such that $\pi(n, n'') = p$;
  the new configuration is then $(n'', B - p)$. 
  Player~1 wins a play of the game when a configuration $(n, 0)$ is
  reached, and he loses (i.e., player~2 wins) when a configuration
  $(n, B)$ is reached in which player~1 is stuck, i.e., for all moves
  $(n, n') \in M$, we have $\pi(n, n') > B$.

  We define the timed automaton 
  $\Tt_G = (L, C, S, A, E, \delta, \rho, F)$ by setting 
  $C = \eset{b, c}$; $S = L \times (\REALS{B_0})^2$; 
  $A = \eset{*} \cup P \cup M$, where $P = \pi(M)$, the image of the
  function $\pi : M \to \Npos$;  
  \begin{eqnarray*}
    L & = & \eset{*} \cup N \cup 
    \Set{(n, p) \: : \: \text{there is } (n, n') \in M,
      \text{s.t. } \pi(n, n') = p};
    \\
    E(a) & = & 
    \begin{cases}
      \set{(n, \nu) \: : \: n \in N \text{ and } \nu(b) = B_0}
      & \text{if $a = *$},
      \\
      \Set{(n, \nu) \: : \: 
        \text{there is } (n, n') \in M, \text{s.t. } \pi(n, n') = p 
        \text{ and } \nu(c) = 0}
      & \text{if $a = p \in P$},
      \\
      \Set{\big((n, p), \nu\big) \: : \:
      \pi(n, n') = p \text{ and } \nu(c) = p} 
      & \text{if $a = (n, n') \in M$},
    \end{cases}
    \\
    \delta(\ell, a) & = & 
    \begin{cases}
      * & \text{if $\ell = n \in N$ and $a = *$},
      \\
      (n, p) & \text{if $\ell = n \in N$ and $a = p \in P$}, 
      \\
      n' & \text{if $\ell = (n, p) \in N \times P$ and 
        $a = (n, n') \in M$}; 
    \end{cases}
  \end{eqnarray*}
  $\rho(a) = \eset{c}$, for every $a \in A$; 
  and $F = \eset{*} \times V$.  
  Note that the timed automaton $\Tt_G$ has only two clocks and that
  the clock $b$ is never reset.

  Finally, we define the reachability game 
  $\Gamma_G = (\Tt_G, L_1, L_2)$ by setting $L_1 = N$ and 
  $L_2 = L \setminus L_1$.
  It is routine to verify that player~1 has a winning strategy from
  state $(n_0, (0, 0)) \in S$ in the reachability game~$\Gamma_G$ if 
  and only if player~1 has a winning strategy 
  (from the initial configuration $(n_0, B_0)$)
  in the countdown game~$G$.  
\end{proof}

\end{document}